\documentclass[conference,letter,romanappendices]{ieeeconf}
\let\proof\relax   

\usepackage{mathpple}
\usepackage{times}
\usepackage{amsthm,xpatch}
\usepackage{amsmath,amsfonts}
\usepackage{cite}
\usepackage{amssymb}
\usepackage{dsfont}
\usepackage{graphicx, subfigure}
\usepackage{color}
\usepackage{breqn}
\usepackage{mathtools}
\usepackage{bbm}
\usepackage{latexsym}
\usepackage[ruled, linesnumbered]{algorithm2e}
\usepackage{accents}
\usepackage{tikz}
\usepackage[mathscr]{euscript}
\usepackage{bm}
\usepackage[frak=lucida]{mathalpha} 
\usepackage{comment}

\usepackage[inline]{trackchanges}

\addeditor{}

\newtheorem{lemma}{Lemma}
\newtheorem{theorem}{Theorem}
\newtheorem{remark}{Remark}

  \renewcommand{\baselinestretch}{0.98}

\newcommand*{\transpose}{%
  {\mathpalette\@transpose{}}%
}

\IEEEoverridecommandlockouts

\begin{document}

\makeatletter
\newcommand{\raisemath}[1]{\mathpalette{\raisem@th{#1}}}
\newcommand{\raisem@th}[3]{\raisebox{#1}{$#2#3$}}
\makeatother

\newcommand{\mstk}{\hspace{-0.145cm}*}

\newcommand{\mstl}{\hspace{-0.105cm}*}

\newcommand{\mstm}{\hspace{-0.175cm}*}

\newcommand{\SB}[3]{
\sum_{#2 \in #1}\biggl|\overline{X}_{#2}\biggr| #3
\biggl|\bigcap_{#2 \notin #1}\overline{X}_{#2}\biggr|
}

\newcommand{\Mod}[1]{\ (\textup{mod}\ #1)}

\newcommand{\overbar}[1]{\mkern 0mu\overline{\mkern-0mu#1\mkern-8.5mu}\mkern 6mu}

\makeatletter
\newcommand*\nss[3]{%
  \begingroup
  \setbox0\hbox{$\m@th\scriptstyle\cramped{#2}$}%
  \setbox2\hbox{$\m@th\scriptstyle#3$}%
  \dimen@=\fontdimen8\textfont3
  \multiply\dimen@ by 4             
  \advance \dimen@ by \ht0
  \advance \dimen@ by -\fontdimen17\textfont2
  \@tempdima=\fontdimen5\textfont2  
  \multiply\@tempdima by 4
  \divide  \@tempdima by 5          
  \ifdim\dimen@<\@tempdima
    \ht0=0pt                        
    \@tempdima=\fontdimen5\textfont2
    \divide\@tempdima by 4          
    \advance \dimen@ by -\@tempdima 
    \ifdim\dimen@>0pt
      \@tempdima=\dp2
      \advance\@tempdima by \dimen@
      \dp2=\@tempdima
    \fi
  \fi
  #1_{\box0}^{\box2}%
  \endgroup
  }
\makeatother

\makeatletter
\renewenvironment{proof}[1][\proofname]{\par
  \pushQED{\qed}%
  \normalfont \topsep6\p@\@plus6\p@\relax
  \trivlist
  \item[\hskip\labelsep
        \itshape
    #1\@addpunct{:}]\ignorespaces
}{%
  \popQED\endtrivlist\@endpefalse
}
\makeatother

\makeatletter
\newsavebox\myboxA
\newsavebox\myboxB
\newlength\mylenA

\newcommand*\xoverline[2][0.75]{%
    \sbox{\myboxA}{$\m@th#2$}%
    \setbox\myboxB\null
    \ht\myboxB=\ht\myboxA%
    \dp\myboxB=\dp\myboxA%
    \wd\myboxB=#1\wd\myboxA
    \sbox\myboxB{$\m@th\overline{\copy\myboxB}$}
    \setlength\mylenA{\the\wd\myboxA}
    \addtolength\mylenA{-\the\wd\myboxB}%
    \ifdim\wd\myboxB<\wd\myboxA%
       \rlap{\hskip 0.5\mylenA\usebox\myboxB}{\usebox\myboxA}%
    \else
        \hskip -0.5\mylenA\rlap{\usebox\myboxA}{\hskip 0.5\mylenA\usebox\myboxB}%
    \fi}
\makeatother

\xpatchcmd{\proof}{\hskip\labelsep}{\hskip3.75\labelsep}{}{}

\pagestyle{plain}

\title{\fontsize{21}{28}\selectfont The Role of Reusable and Single-Use Side Information in Private Information Retrieval}

\author{Anoosheh Heidarzadeh and Alex Sprintson\thanks{The authors are with the Department of Electrical and Computer Engineering, Texas A\&M University, College Station, TX 77843 USA (E-mail: \{anoosheh, spalex\}@tamu.edu).}
}

%

%


\maketitle 

\thispagestyle{plain}

\begin{abstract}
This paper introduces the problem of Private Information Retrieval with Reusable and Single-use Side Information (PIR-RSSI). 
In this problem, one or more remote servers store identical copies of a set of $K$ messages, and there is a user that initially knows $M$ of these messages, and wants to privately retrieve one other message from the set of $K$ messages.
The objective is to design a retrieval scheme in which the user downloads the minimum amount of information from the server(s) while the identity of the message wanted by the user and the identities of an $M_1$-subset of the $M$ messages known by the user (referred to as reusable side information) are protected, but the identities of the remaining $M_2=M-M_1$ messages known by the user (referred to as single-use side information) do not need to be protected. 
The PIR-RSSI problem reduces to the classical Private Information Retrieval (PIR) problem when ${M_1=M_2=0}$, and reduces to the problem of PIR with Private Side Information or PIR with Side Information when ${M_1\geq 1,M_2=0}$ or ${M_1=0,M_2\geq 1}$, respectively. 
In this work, we focus on the single-server setting of the PIR-RSSI problem. 
We characterize the capacity of this setting for the cases of ${M_1=1,M_2\geq 1}$ and ${M_1\geq 1,M_2=1}$, where the capacity is defined as the maximum achievable download rate over all PIR-RSSI schemes. 
Our results show that for sufficiently small values of $K$, the single-use side information messages can help in reducing the download cost only if they are kept private; and for larger values of $K$, the reusable side information messages cannot help in reducing the download cost.

\end{abstract}

\section{Introduction}
In this work, we introduce the problem of \emph{Private Information Retrieval with Reusable and Single-use Side Information (PIR-RSSI)}.
In this problem, there is a single (or multiple) remote server(s) storing (identical copies of) $K$ messages, and there is a user that is interested in privately retrieving one message from the set of $K$ messages. 
The user has a prior side information about a subset of messages. 
In particular, the user initially knows $M$ messages from the set of $K$ messages, distinct from the message required by the user, and the servers do not initially know which $M$-subset of messages is known by the user. 
The goal of the user is to privately retrieve their desired message by downloading the minimum amount of information from the server(s) 
while protecting the identities of some of their side information messages. 
In particular, the identities of  $M_1$ (out of $M$) side information messages must be protected, 
whereas the identities of the remaining $M_2:=M-M_1$ side information messages do not need to be protected.
We refer to the $M_1$ side information messages whose identities need to be protected as \emph{reusable side information (RSI)}, and refer to the remaining $M_2$ side information messages as \emph{single-use side information (SSI)}.   

The PIR-RSSI problem is motivated by several practical scenarios. 
For example, the user may want to use some part of their side information to minimize the download cost in one round of retrieval while keeping the other part of their side information reusable for future rounds. 
As another example, consider a scenario in which the server stores a dataset that contains information about individuals including the user themselves, and the user's side information consists of their own data and the data pertaining to some other individuals. 
In this scenario, the identity of the user's own data---which is part of the user's side information---needs to be kept private in order to protect the user's identity.

The PIR-RSSI problem is a generalization of the classical Private Information Retrieval (PIR) problem~\cite{SJ2017,TSC2019}, the problem of PIR with Private Side Information (PIR-PSI)~\cite{KGHERS2020,CWJ2020}, and the problem of PIR with Side Information (PIR-SI)~\cite{KGHERS2020,LG2020CISS}.
In particular, the PIR-RSSI problem with ${M_1=0,M_2=0}$ is equivalent to the classical PIR problem; the PIR-RSSI problem with $M_1\geq 1, M_2=0$ is equivalent to the PIR-PSI problem with $M_1$ side information messages, and the PIR-RSSI problem with $M_1=0,M_2\geq 1$ is equivalent to the PIR-SI problem with $M_2$ side information messages. 
This is because in the PIR-PSI problem, the identities of all side information messages must be kept private, and in the PIR-SI problem, the identity of none of the side information messages needs to be kept private.
Similar to PIR-PSI and PIR-SI, several variants of the PIR-RSSI problem can also be considered, e.g., settings in which the user wishes to retrieve more than one message (see, e.g.,~\cite{BU17,BU2018,HKGRS2018,SSM2018,LG2018,HKRS2019,HS2021}), or the user's side information consists of some coded combinations or uncoded fractions of a subset of messages (see, e.g.,~\cite{T2017,WBU2018,WBU2018No2,HKS2018,HKS2019,HKS2019Journal,KKHS12019,KKHS22019}), or the servers store coded versions of the messages (see, e.g.,~\cite{TER2017,TGKHHER2017}).
In this work, we consider the setting in which the user wants to retrieve only one message and the user's side information is a subset of messages. 
In addition, we focus on the single-server setting, and only briefly discuss the multi-server setting.  


As shown in~\cite{KGHERS2020}, the capacity of single-server PIR-PSI and PIR-SI with $K$ messages and $M$ side information messages is given by ${1/(K-M)}$ and ${1/\lceil K/(M+1) \rceil}$, respectively, where the capacity is defined as the maximum achievable download rate.  
These results were also extended to the multi-server setting of PIR-PSI and PIR-SI in~\cite{CWJ2020} and~\cite{LG2020CISS}, respectively.
By combining the converse proof techniques of~\cite{KGHERS2020} for single-server PIR-PSI and PIR-SI, it can be shown that the capacity of single-server PIR-RSSI is upper bounded by ${1/\lceil (K-M_1)/(M_2+1)\rceil}$.
Note that this converse bound is tight for the cases of ${M_1\geq 1, M_2=0}$ (PIR-PSI) and  ${M_1=0,M_2\geq 1}$ (PIR-SI).
Thus, a natural question that arises is whether this bound is tight for any ${M_1\geq 1,M_2\geq 1}$. 
A simple comparison of this bound and the capacity of single-server PIR-PSI and PIR-SI shows that if this bound was tight, then the $M_1$ RSI messages and the $M_2$ SSI messages could be leveraged to the full extent of their individual potential in reducing the download cost. 

In this work, we characterize the capacity of single-server PIR-RSSI for the cases of ${M_1=1,M_2\geq 1}$ and ${M_1\geq 1,M_2=1}$, and 
show that the converse bound ${1/\lceil (K-M_1)/(M_2+1)\rceil}$ is not tight in general.
In particular, we prove that the capacity of single-server PIR-RSSI is upper bounded by ${1/\min\{K-M_2-1,\lceil K/(M_2+1)\rceil\}}$ or ${1/\min\{K-M_1-1,\lceil K/2\rceil\}}$ when ${M_1=1}$ or ${M_2=1}$, respectively.
Our converse proofs are based on information-theoretic and combinatorial arguments that rely on a necessary condition for any PIR-RSSI scheme.  
In addition, we build up on the existing schemes for single-server PIR-PSI and PIR-SI, and prove the achievability of the rate ${1/\min\{K-M_1-M_2,\lceil K/(M_2+1)\rceil\}}$ for all ${M_1\geq 1,M_2\geq 1}$. 
Our results show that for sufficiently small values of $K$, i.e., any $K$ such that ${K-M_1-M_2<\lceil K/(M_2+1) \rceil}$, 
the $M_2$ SSI messages can help in reducing the download cost only if they are kept private; and 
for larger values of $K$, i.e., any $K$ such that ${K-M_1-M_2\geq\lceil K/(M_2+1) \rceil}$, the $M_1$ RSI messages cannot help in reducing the download cost.


\section{Problem Setup}\label{sec:SN}
Throughout, random variables and their realizations are denoted by bold-face symbols and regular symbols, respectively.
For any integer $i\geq 1$, we denote  ${\{1,\dots,i\}}$ by $[i]$. 
 
Let $\mathbbmss{F}_q$ be a finite field of order $q$, and let $\mathbbmss{F}_{q}^{n}$ be the $n$-dimensional vector space over $\mathbbmss{F}_q$. 
Let $K,M_1,M_2$ be arbitrary integers such that $M_1,M_2\geq 1$ and ${K> M_1+M_2}$. 

Consider a server that stores $K$ messages ${\mathrm{X}_1,\dots,\mathrm{X}_K}$, where $\mathrm{X}_k\in \mathbbmss{F}_q^{n}$ for $k\in [K]$. 
We denote by $\mathrm{X}_{\mathrm{T}}$ the set of messages $\{\mathrm{X}_k: k\in \mathrm{T}\}$ for every ${\mathrm{T}\subseteq [K]}$. 

Let $\mathcal{R}$ and $\mathcal{S}$ be the set of all $M_1$-subsets and all $M_2$-subsets of $[K]$, respectively.  
Consider a user who initially knows the $M_1+M_2$ messages $\mathrm{X}_{\mathrm{R}}\cup \mathrm{X}_{\mathrm{S}}$ for a given ${(\mathrm{R},\mathrm{S})\in \mathcal{R}\times \mathcal{S}}$ such that $\mathrm{R}\cap\mathrm{S} = \emptyset$, and wishes to retrieve the message $\mathrm{X}_{\mathrm{W}}$ for a given $\mathrm{W}\in [K]\setminus (\mathrm{R}\cup\mathrm{S})$.
We refer to $\mathrm{X}_{\mathrm{W}}$ as the \emph{demand message}, 
$\mathrm{X}_{\mathrm{R}}$ as the \emph{reusable side information (RSI) message(s)}, $\mathrm{X}_{\mathrm{S}}$ as the \emph{single-use side information (SSI) message(s)}, $\mathrm{W}$ as the \emph{demand index}, $\mathrm{R}$ as the \emph{index set of the RSI}, $\mathrm{S}$ as the \emph{index set of the SSI}, $M_1$ as the \emph{size of the RSI}, and $M_2$ as the \emph{size of the SSI}.

In this work, we assume that:  
\begin{enumerate}
\item $\mathbf{X}_1,\dots,\mathbf{X}_K$ are independent and uniformly distributed over $\mathbbmss{F}_{q}^{n}$. 
Thus, ${H(\mathbf{X}_{\mathrm{T}})= |\mathrm{T}| B}$ for all ${\mathrm{T}\subseteq [K]}$, where $B:= n\log_2 q$ is the entropy of a message. 
\item $(\mathbf{W},\mathbf{R},\mathbf{S})$ and $\mathbf{X}_{1},\dots,\mathbf{X}_K$ are independent. 
\item The distribution of ${(\mathbf{R},\mathbf{S})}$ is uniform over all ${(\mathrm{R},\mathrm{S})\in \mathbbm{R}\times\mathbbm{S}}$ such that ${\mathrm{R}\cap\mathrm{S} = \emptyset}$, and the conditional distribution of $\mathbf{W}$ given $(\mathbf{R},\mathbf{S})=(\mathrm{R},\mathrm{S})$ is uniform over all $\mathrm{W}\in [K]\setminus (\mathrm{R}\cup\mathrm{S})$. 
\item The size of the RSI ($M_1$), the size of the SSI ($M_2$), and the distribution of $(\mathbf{W},\mathbf{R},\mathbf{S})$ are initially known by the server, whereas the realization $(\mathrm{W},\mathrm{R},\mathrm{S})$ is initially unknown to the server.
\end{enumerate}

Given $(\mathrm{W},\mathrm{R},\mathrm{S})$, the user generates a query $\mathrm{Q}^{[\mathrm{W},\mathrm{R},\mathrm{S}]}$, simply denoted by $\mathrm{Q}$, and sends it to the server. 
The query $\mathrm{Q}$ is a deterministic or stochastic function of $(\mathrm{W},\mathrm{R},\mathrm{S})$, independent of $\mathrm{X}_1,\dots,\mathrm{X}_K$.
For the ease of notation, we denote $\mathbf{Q}^{[\mathbf{W},\mathbf{R},\mathbf{S}]}$ by $\mathbf{Q}$. 
The query $\mathrm{Q}$ must satisfy the following condition: 
${\mathbb{P}(\mathbf{W}=\mathrm{W}^{*}, \mathbf{R} = \mathrm{R}^{*}|\mathbf{Q}=\mathrm{Q})=\mathbb{P}(\mathbf{W}=\mathrm{W}^{*}, \mathbf{R} = \mathrm{R}^{*})}$ for all $(\mathrm{W}^{*},\mathrm{R}^{*})\in [K]\times \mathcal{R}$ such that $\mathrm{W}^{*}\not\in \mathrm{R}^{*}$. 
This condition is referred to as the \emph{privacy condition}. 
The privacy condition ensures 
(i) the privacy of the demand index, i.e., ${\mathbb{P}(\mathbf{W}=\mathrm{W}^{*}|\mathbf{Q}=\mathrm{Q})=\mathbb{P}(\mathbf{W}=\mathrm{W}^{*})}$ for all $\mathrm{W}^{*}\in [K]$, and 
(ii) the privacy of the RSI index set (and hence, the name ``reusable side information''), i.e., ${\mathbb{P}(\mathbf{R}=\mathrm{R}^{*}|\mathbf{Q}=\mathrm{Q})=\mathbb{P}(\mathbf{R}=\mathrm{R}^{*})}$ for all $\mathrm{R}^{*}\in \mathcal{R}$. 
Note that the privacy condition does not guarantee the privacy of the SSI index set (and hence, the name ``single-use side information''), i.e., ${\mathbb{P}(\mathbf{S} = \mathrm{S}^{*}|\mathbf{Q}=\mathrm{Q})}$ and ${\mathbb{P}(\mathbf{S} = \mathrm{S}^{*})}$ are not necessarily equal for all $\mathrm{S}^{*}\in \mathcal{S}$. 

Upon receiving $\mathrm{Q}$, the server generates an answer $\mathrm{A}^{[\mathrm{W},\mathrm{R},\mathrm{S}]}$, simply denoted by $\mathrm{A}$, and sends it back to the user. 
The answer $\mathrm{A}$ is a deterministic function of $\mathrm{Q}$ and $\mathrm{X}_1,\dots,\mathrm{X}_K$. 
For simplifying the notation, we denote $\mathbf{A}^{[\mathbf{W},\mathbf{R},\mathbf{S}]}$ by $\mathbf{A}$. 
Thus, ${H(\mathbf{A}|\mathbf{Q},\mathbf{X}_1,\dots,\mathbf{X}_K)=0}$. 
The user must be able to recover  $\mathrm{X}_{\mathrm{W}}$ given $\mathrm{A}$, $\mathrm{Q}$, $\mathrm{X}_{\mathrm{R}}$, $\mathrm{X}_{\mathrm{S}}$, and $(\mathrm{W}, \mathrm{R}, \mathrm{S})$. 
That is,
$H(\mathbf{X}_{\mathrm{W}}| \mathbf{A},\mathbf{Q},\mathbf{X}_{\mathrm{R}},\mathbf{X}_{\mathrm{S}})=0$. 
This condition is referred to as the \emph{recoverability condition}. 

The problem is to design a protocol for generating a query $\mathrm{Q}^{[\mathrm{W},\mathrm{R},\mathrm{S}]}$ and the corresponding answer $\mathrm{A}^{[\mathrm{W},\mathrm{R},\mathrm{S}]}$ for any given $(\mathrm{W},\mathrm{R},\mathrm{S})$
such that both the privacy and recoverability conditions are satisfied.
We refer to this problem as \emph{single-server Private Information Retrieval with Reusable and Single-use Side Information}, or \emph{PIR-RSSI} for short.

The \emph{rate} of a PIR-RSSI protocol is defined as the ratio of the amount of information required by the user (i.e., $H(\mathbf{X}_{\mathbf{W}})$) to the amount of information downloaded from the server (i.e., $H(\mathbf{A}^{[\mathbf{W},\mathbf{R},\mathbf{S}]})$), and the \emph{capacity} of PIR-RSSI is defined as the supremum of rates over all PIR-RSSI protocols.

Our goal in this work is to characterize the capacity of PIR-RSSI in terms of the parameters $K,M_1,M_2$. 

\section{Main Results}
In this section, we present our main results. 
Theorem~\ref{thm:Conv} provides an upper bound on the capacity for ${M_1=1, M_2\geq 1}$ and ${M_1\geq 1, M_2=1}$, and Theorem~\ref{thm:Conv} provides a lower bound on the capacity for any ${M_1\geq 1,M_2\geq 1}$. 
The proof of converse and achievability are presented in Sections~\ref{sec:Conv} and~\ref{sec:Ach}, respectively.    

\begin{theorem}\label{thm:Conv}
For PIR-RSSI with $K$ messages, RSI's size $M_1$, and SSI's size $M_2$, the capacity is upper bounded by ${1/\min\{K-M_2-1,\lceil K/(M_2+1)\rceil\}}$ or ${1/\min\{K-M_1-1,\lceil K/2\rceil\}}$ when ${M_1=1}$ or ${M_2=1}$, respectively.
\end{theorem}

To prove this result, we use a mix of information-theoretic and combinatorial arguments. 
These arguments rely on a necessary condition imposed by the privacy and recoverability conditions. 
For the case of $M_1=1$ or $M_2=1$, we show that for any PIR-RSSI protocol there exists a subset of messages of size at most ${\max\{M_1+M_2, \lfloor K M_2/(M_2+1)\rfloor\}}$ given which all $K$ messages can be recovered from the query and the answer. 
This result yields a lower bound of ${(K-\max\{M_1+M_2, \lfloor KM_2/(M_2+1)\rfloor \})B}$, or subsequently, ${(\min\{K-M_1-M_2, \lceil K/(M_2+1)\rceil\})B}$, on the amount of information downloaded from the server, where $B$ is the amount of information in a message. 
(Note that ${K-\lfloor K M_2/(M_2+1)\rfloor\geq \lceil K/(M_2+1)\rceil}$.)
Thus, the rate of any PIR-RSSI protocol for ${M_1=1}$ or ${M_2=1}$ is upper bounded by ${1/\min\{K-M_1-M_2,\lceil K/(M_2+1)\rceil\}}$. 

\begin{theorem}\label{thm:Ach}
For PIR-RSSI with $K$ messages, RSI's size $M_1$, and SSI's size $M_2$, the capacity is lower bounded by ${1/\min\{K-M_1-M_2,\lceil K/(M_2+1)\rceil\}}$.
\end{theorem}

The proof is based on a simple modification of the existing schemes for single-server PIR-PSI and single-server PIR-SI. 
In particular, the rate ${1/(K-M_1-M_2)}$ is achievable by a modified version of the MDS Code scheme of~\cite{KGHERS2020}; 
and the rate ${1/\lceil K/(M_2+1) \rceil}$ is achievable by a modified version of the Partition-and-Code scheme of~\cite{KGHERS2020}. 

\begin{remark}\label{rem:Comp}
\emph{Using similar proof techniques as in~\cite{KGHERS2020}, it can be shown that the capacity of PIR-RSSI is upper bounded by ${1/\lceil (K-M_1)/(M_2+1)\rceil}$. 
This bound suggests that the $M_1$ RSI messages can be leveraged to reduce the effective number of messages from $K$ to ${K-M_1}$, and the $M_2$ SSI messages can be further leveraged to reduce the effective number of messages from ${K-M_1}$ to ${\lceil (K-M_1)/(M_2+1)\rceil}$.
The tightness of this bound for ${M_1=0,M_2\geq 1}$ and ${M_1\geq 1,M_2=0}$ follows from the results in~\cite{KGHERS2020} for single-server PIR-SI and single-server PIR-PSI, respectively. 
Interestingly, our results in this work disprove the tightness of this bound for ${M_1=1,M_2\geq 1}$ and ${M_1\geq 1, M_2=1}$. 
This shows that the RSI and the SSI cannot always be leveraged simultaneously so as to reduce the effective number of messages from $K$ to ${\lceil (K-M_1)/(M_2+1)\rceil}$.}
\end{remark}

\begin{remark}\label{rem:Conj}
\emph{We conjecture that the capacity of PIR-RSSI is given by ${1/\min\{K-M_1-M_2,\lceil K/(M_2+1)\rceil\}}$ for all $K,M_1,M_2$.
The correctness of this conjecture for ${M_1=0,M_2\geq 1}$ and ${M_1\geq 1, M_2=0}$ follows from the results of~\cite{KGHERS2020}; and  Theorems~\ref{thm:Conv} and~\ref{thm:Ach} prove this conjecture for ${M_1=1,M_2\geq 1}$ and ${M_1\geq 1,M_2=1}$. 
The proof of this conjecture remains open in general for ${M_1>1, M_2>1}$. 
In addition, we conjecture that the capacity of multi-server PIR-RSSI is given by ${(1-1/N^{K^{*}})/(1-1/N)}$, where $N$ is the number of servers each of which stores an identical copy of the $K$ messages, and ${K^{*}:= \min\{K-M_1-M_2,\lceil K/(M_2+1)\rceil\}}$. 
The achievability of this rate follows from a simple extension of the existing schemes for multi-server PIR-PSI~\cite{KGHERS2020,CWJ2020} and multi-server PIR-SI~\cite{KGHERS2020,LG2020CISS}. 
The proof of converse for ${M_1=0,M_2\geq 1}$ and ${M_1\geq 1, M_2=0}$ follows from the results of~\cite{CWJ2020} and~\cite{LG2020CISS}, and remains open for ${M_1\geq 1, M_2\geq 1}$.}
\end{remark}

\section{Proof of Theorem~\ref{thm:Conv}}\label{sec:Conv}
We present the proof for the following cases separately: 
(i) ${M_1=1,M_2=1}$; 
(ii) ${M_1=1,M_2>1}$; and 
(iii) ${M_1>1,M_2=1}$. 
In each case, 
for sufficiently small values of $K$, i.e., any $K$ such that ${K-M_1-M_2<\lceil K/(M_2+1)\rceil}$, we need to show that ${H(\mathbf{A})\geq (K-M_1-M_2)B}$, and for sufficiently large values of $K$, i.e., any $K$ such that ${K-M_1-M_2\geq \lceil K/(M_2+1)\rceil}$, we need to show that ${H(\mathbf{A})\geq \lceil K/(M_2+1)\rceil B}$, where $B$ is the entropy of a message.
For sufficiently small (or large) $K$, 
it suffices to show that there exist ${L\leq M_1+M_2}$ (or ${L\leq \lfloor KM_2/(M_2+1) \rfloor}$) messages $\mathrm{X}_{\mathrm{I}}$ for some $L$-subset $\mathrm{I}$ of $[K]$ given which all other $K-L$ messages $\mathrm{X}_{[K]\setminus \mathrm{I}}$ can be recovered from the query and the answer, i.e., $H(\mathbf{X}_{[K]\setminus \mathrm{I}}|\mathbf{A},\mathbf{Q},\mathbf{X}_{\mathrm{I}}) = 0$.
This is because $H(\mathbf{A})\geq H(\mathbf{A}|\mathbf{Q}) \geq H(\mathbf{A}|\mathbf{Q},\mathbf{X}_{\mathrm{I}})$, 
and 
$H(\mathbf{A}|\mathbf{Q},\mathbf{X}_{\mathrm{I}}) 
\stackrel{\tiny\text{(a)}}{=} H(\mathbf{A}|\mathbf{Q},\mathbf{X}_{\mathrm{I}}) + H(\mathbf{X}_{[K]\setminus \mathrm{I}}|\mathbf{A},\mathbf{Q},\mathbf{X}_{\mathrm{I}}) 
\stackrel{\tiny\text{(b)}}{=} H(\mathbf{X}_{[K]\setminus \mathrm{I}}|\mathbf{Q},\mathbf{X}_{\mathrm{I}}) + H(\mathbf{A}|\mathbf{Q},\mathbf{X}_{\mathrm{I}},\mathbf{X}_{[K]\setminus \mathrm{I}}) 
\stackrel{\tiny\text{(c)}}{=} H(\mathbf{X}_{[K]\setminus \mathrm{I}}|\mathbf{Q},\mathbf{X}_{\mathrm{I}}) 
\stackrel{\tiny\text{(d)}}{=} H(\mathbf{X}_{[K]\setminus \mathrm{I}}) 
= (K-L)B$, where 
(a) holds because ${H(\mathbf{X}_{[K]\setminus \mathrm{I}}|\mathbf{A},\mathbf{Q},\mathbf{X}_{\mathrm{I}}) = 0}$ by assumption; 
(b) follows from the chain rule of entropy; 
(c) holds because $H(\mathbf{A}|\mathbf{Q},\mathbf{X}_{\mathrm{I}},\mathbf{X}_{[K]\setminus \mathrm{I}}) = 0$; and 
(d) follows from the independence of $\mathbf{Q}$ and $\mathbf{X}_1,\dots,\mathbf{X}_K$. 
Thus, $H(\mathbf{A})\geq (K-L)B\geq (K-M_1-M_2)B$ for ${L\leq M_1+M_2}$, and $H(\mathbf{A})\geq (K-L)B\geq (K-\lfloor KM_2/(M_2+1) \rfloor)B \geq \lceil K/(M_2+1)\rceil B$ for ${L\leq \lfloor KM_2/(M_2+1) \rfloor}$.

The following two lemmas are useful in the proofs. 

\begin{lemma}\label{lem:NC}
Given any PIR-RS-SI protocol, for any ${\mathrm{W}^{*}\in [K]}$ and ${\mathrm{R}^{*}\in \mathcal{R}}$ such that ${\mathrm{W}^{*}\not\in \mathrm{R}^{*}}$, there must exist ${\mathrm{S}^{*}\in \mathcal{S}}$ such that ${\mathrm{W}^{*}\not\in \mathrm{S}^{*}}$ and ${\mathrm{R}^{*}\cap \mathrm{S}^{*} = \emptyset}$, such that $\mathrm{X}_{\mathrm{W}^{*}}$ can be recovered from the query and the answer given ${\mathrm{X}_{\mathrm{R}^{*}}\cup \mathrm{X}_{\mathrm{S}^{*}}}$, 
i.e., ${H(\mathbf{X}_{\mathrm{W}^{*}}|\mathbf{A},\mathbf{Q},\mathbf{X}_{\mathrm{R}^{*}},\mathbf{X}_{\mathrm{S}^{*}})=0}$. 
\end{lemma}

\begin{proof}
The proof is straightforward by the way of contradiction---relying on the privacy and recoverability conditions, and hence omitted for brevity. 
\end{proof}

\begin{lemma}\label{lem:Grow}
Given any PIR-RS-SI protocol, if there exist ${\alpha\geq \lceil M_1 M_2/(M_2+1)\rceil}$ messages given which ${\beta\geq \lceil\alpha/M_2\rceil}$ other messages can be recovered from the query and the answer, there exist ${L\leq \lfloor K M_2/(M_2+1)\rfloor}$ messages given which all other ${K-L}$ messages can be recovered from the query and the answer. 
\end{lemma}

\begin{proof}
Suppose that given the messages  ${\mathrm{X}_{\mathrm{I}}}$ for some $\alpha$-subset $\mathrm{I}$ of $[K]$ the messages ${\mathrm{X}_{\mathrm{J}}}$ for some $\beta$-subset $\mathrm{J}$ of $[K]\setminus \mathrm{I}$ can be recovered. 
Note that ${|\mathrm{I}\cup\mathrm{J}| = \alpha+\beta\geq M_1}$ since ${\alpha\geq M_1M_2/(M_2+1)}$ and ${\beta\geq \alpha/M_2\geq M_1/(M_2+1)}$ by assumption. 
Fix an arbitrary ${\mathrm{W}^{*}\in [K]\setminus (\mathrm{I}\cup \mathrm{J})}$ and an arbitrary $M_1$-subset $\mathrm{R}^{*}$ of $\mathrm{I}\cup\mathrm{J}$. 
By Lemma~\ref{lem:NC}, there exists an $M_2$-subset $\mathrm{S}^{*}$ of ${[K]\setminus (\{\mathrm{W}^{*}\}\cup \mathrm{R}^{*})}$ such that $\mathrm{X}_{\mathrm{W}^{*}}$ can be recovered given $\mathrm{X}_{\mathrm{R}^{*}}\cup \mathrm{X}_{\mathrm{S}^{*}}$. 
Let ${\mathrm{T}:=\mathrm{S}^{*}\setminus (\mathrm{I}\cup \mathrm{J})}$. 
(Note that ${|\mathrm{T}|\leq |\mathrm{S}^{*}| = M_2}$.)
Observe that the messages ${\mathrm{X}_{\mathrm{J}} \cup \mathrm{X}_{\mathrm{W}^{*}}}$ can be recovered given the messages ${\mathrm{X}_{\mathrm{I}}\cup\mathrm{X}_{\mathrm{T}}}$. 
Note that ${|\mathrm{J}\cup\{\mathrm{W}^{*}\}| = \beta+1}$ and ${|\mathrm{I}\cup \mathrm{T}| = \alpha+|\mathrm{T}|}$. 
Thus, there exist $\alpha^{*}:=\alpha+|\mathrm{T}|$ messages given which $\beta^{*}:=\beta+1$ other messages can be recovered. 
Note that ${\beta^{*}\geq \lceil\alpha^{*}/M_2\rceil}$. 
This is because $\beta^{*}=\beta+1\geq {\lceil\alpha/M_2\rceil+1}\geq {\lceil\alpha/M_2\rceil+\lceil|\mathrm{T}|/M_2\rceil} = \lceil\alpha/M_2+|\mathrm{T}|/M_2\rceil = \lceil\alpha^{*}/M_2\rceil$. 
Thus, there exist $\alpha^{*}$ messages given which ${\beta^{*}\geq \lceil \alpha^{*}/M_2\rceil}$ other messages can be recovered. 
By repeating this argument, it can be shown that there exist $L$ messages given which all other $K-L$ ($\geq \lceil L/M_2\rceil$) messages can be recovered from the query and the answer.
Solving the inequality ${K-L\geq \lceil L/M_2\rceil}$ for $L$, it follows that ${L\leq \lfloor KM_2/(M_2+1)\rfloor}$, as was to be shown. 
\end{proof}

Consider an arbitrary PIR-RSSI protocol. 
Let $\mathrm{Q}$ and $\mathrm{A}$ be a query and its corresponding answer generated by the protocol, respectively. 
For any ${\mathrm{W}^{*}\in [K]}$, any ${\mathrm{R}^{*}\in \mathcal{R}}$ such that ${\mathrm{W}^{*}\not\in \mathrm{R}^{*}}$, and any ${\mathrm{S}^{*}\in \mathcal{S}}$ such that ${\mathrm{W}^{*}\not\in \mathrm{S}^{*}}$ and ${\mathrm{R}^{*}\cap \mathrm{S}^{*} = \emptyset}$, we define a tuple $(\mathrm{R}^{*},\mathrm{S}^{*},\mathrm{W}^{*})$ if $\mathrm{X}_{\mathrm{W}^{*}}$ can be recovered from $\mathrm{Q}$ and $\mathrm{A}$ given ${\mathrm{X}_{\mathrm{R}^{*}}\cup \mathrm{X}_{\mathrm{S}^{*}}}$. 
Note that the result of Lemma~\ref{lem:NC} implies that for any $\mathrm{W}^{*}$ and any $\mathrm{R}^{*}$, there exists a tuple $(\mathrm{R}^{*},\mathrm{S}^{*},\mathrm{W}^{*})$ for some $\mathrm{S}^{*}$. 

Throughout the proof, we focus mainly on an arbitrary (but fixed) collection of tuples, $(\mathrm{R}_1,\mathrm{S}_1,\mathrm{W}_1),\dots, (\mathrm{R}_T,\mathrm{S}_T,\mathrm{W}_T)$, for some integer ${T\geq 1}$ (depending on $\mathrm{Q}$ and $\mathrm{A}$), satisfying the following conditions: 
(i) ${\mathrm{R}_i = [M_1]}$ for all ${i\in [T]}$; 
(ii) ${\mathrm{W}_i\not\in [M_1]}$ for all ${i\in [T]}$; 
(iii) ${\mathrm{W}_i\not\in \cup_{j=1}^{i-1} (\mathrm{S}_j \cup \mathrm{W}_j)}$ for all ${i\in [T]}$; and 
(iv) ${\cup_{j=1}^{T} (\mathrm{S}_j\cup \mathrm{W}_j)= [K]\setminus [M_1]}$.
Note that Lemma~\ref{lem:NC} guarantees the existence of such a collection of tuples. 
For simplifying the notation, for every $\mathrm{T}\subseteq [T]$ we denote $\cup_{i\in\mathrm{T}}\mathrm{S}_i$ and $\cup_{i\in\mathrm{T}}\mathrm{W}_i$ by $\mathrm{S}_{\mathrm{T}}$ and $\mathrm{W}_{\mathrm{T}}$, respectively.  

We use the notation $\mathrm{I}\rightarrow\mathrm{J}$ to represent that the messages $\mathrm{X}_{\mathrm{J}}$ can be recovered from $\mathrm{Q}$ and $\mathrm{A}$ given the messages $\mathrm{X}_{\mathrm{I}}$.
With an abuse of notation, we denote the set $\{\mathrm{W}\}$ by $\mathrm{W}$.

\begin{theorem}\label{thm:11}
For PIR-RSSI with $K$ messages, RSI's size ${M_1=1}$, and SSI's size ${M_2=1}$, the capacity is upper bounded by ${1/\min\{K-2,\lceil K/2 \rceil\}}$. 
\end{theorem}

\begin{proof}
In this case, $K> M_1+M_2=2$.
Note that for ${K=3}$ (or $K\geq 4$), ${\min\{K-2,\lceil K/2 \rceil\}}$ is given by $K-2$ (or ${\lceil K/2 \rceil}$).
Thus, for $K=3$ (or $K\geq 4$), we need to show that there exist $L\leq 2$ (or $L\leq \lfloor K/2 \rfloor$) messages given which all other $K-L$ messages can be recovered.  

\emph{$K=3$:} 
Taking $\mathrm{R}_1=\{1\}$ and $\mathrm{W}_1=\{3\}$, it is obvious that $\mathrm{S}_1 = \{2\}$, i.e., the tuple $(\{1\},\{2\},\{3\})$ exists. 
Thus, $\{1,2\}\rightarrow \{3\}$, i.e., given the $2$ messages $\mathrm{X}_1,\mathrm{X}_2$ the other message $\mathrm{X}_3$ can be recovered. 


\emph{$K\geq 4$:} 
By Lemma~\ref{lem:Grow}, it suffices to show that there exist ${\alpha\geq \lceil M_1M_2/(M_2+1)\rceil =1}$ messages given which ${\beta\geq \lceil \alpha/M_2\rceil=\alpha}$ other messages can be recovered. 

Let ${\{(\mathrm{R}_i,\mathrm{S}_i,\mathrm{W}_i)\}_{i\in [T]}}$ be a collection of tuples as defined earlier. 
Note that ${\mathrm{R}_i = \{1\}}$ for all ${i\in [T]}$.
There are two cases: 
(i) ${\mathrm{S}_j = \mathrm{S}_i}$ or ${\mathrm{S}_j = \mathrm{W}_i}$ for some ${1\leq i<j\leq T}$; and
(ii) ${\mathrm{S}_j\neq \mathrm{S}_i}$ and ${\mathrm{S}_j\neq \mathrm{W}_i}$ for all ${1\leq i<j\leq T}$. 

First, consider the case (i). 
In this case, $\{1\}\cup\mathrm{S}_i\rightarrow \mathrm{W}_i\cup \mathrm{W}_j$. 
This is because ${\{1\}\cup\mathrm{S}_i\rightarrow \mathrm{W}_i}$, ${\{1\}\cup\mathrm{S}_j\rightarrow \mathrm{W}_j}$, and ${\mathrm{S}_j = \mathrm{S}_i}$ or $\mathrm{S}_j = \mathrm{W}_i$ (by assumption).
Note that ${\{1\}\cup\mathrm{S}_i}$ and ${\mathrm{W}_i\cup \mathrm{W}_j}$ are disjoint. 
Moreover, ${|\{1\}\cup\mathrm{S}_i| = 2}$ and ${|\mathrm{W}_i\cup \mathrm{W}_j|=2}$. 
Thus, given $\alpha=2$ messages indexed by $\{1\}\cup \mathrm{S}_i$, $\beta=2$ other messages indexed by  ${\mathrm{W}_i}\cup{\mathrm{W}_j}$ can be recovered. 


Now, consider the case (ii). 
Let ${\mathrm{W}^{*}_k = \mathrm{S}_k}$ for any ${k\in [T]}$, and let ${\mathrm{R}_{k,l}^{*} = \mathrm{W}_l}$ for any ${k,l\in [T]}$. 
By Lemma~\ref{lem:NC}, there exists some $\mathrm{S}_{k,l}^{*}$ such that the tuple $(\mathrm{R}_{k,l}^{*},\mathrm{S}_{k,l}^{*},\mathrm{W}_k^{*})$ exists. 
That is, ${\mathrm{R}_{k,l}^{*}\cup \mathrm{S}_{k,l}^{*}\rightarrow \mathrm{W}_k^{*}}$, or equivalently, $\mathrm{W}_l\cup \mathrm{S}_{k,l}^{*}\rightarrow \mathrm{S}_k$. 
First, suppose that $\mathrm{S}_{k,l}^{*}\neq \mathrm{W}_k$ for some $k,l\in [T]$. 
In this case, ${\mathrm{S}_{k,l}^{*}\subset \{1\}\cup \mathrm{S}_{[T]\setminus \{k\}}\cup\mathrm{W}_{[T]\setminus \{k\}}}$.
Note that ${\{1\}\cup \mathrm{S}_{[T]\setminus \{k\}}\rightarrow \mathrm{W}_{[T]\setminus \{k\}}}$.
Recall that ${\mathrm{W}_l\cup \mathrm{S}_{k,l}^{*}\rightarrow \mathrm{S}_k}$.
Since ${\mathrm{W}_l\subset \mathrm{W}_{[T]\setminus \{k\}}}$ and ${\mathrm{S}_{k,l}^{*}\subset \{1\}\cup \mathrm{S}_{[T]\setminus \{k\}}\cup\mathrm{W}_{[T]\setminus \{k\}}}$, then ${\{1\}\cup \mathrm{S}_{[T]\setminus \{k\}}\rightarrow \mathrm{W}_{[T]\setminus \{k\}}}\cup \mathrm{S}_k$. 
Since ${\{1\}\cup \mathrm{S}_{[T]\setminus \{k\}}}$ and ${\mathrm{W}_{[T]\setminus \{k\}}\cup \mathrm{S}_k}$ are disjoint, and ${|\{1\}\cup \mathrm{S}_{[T]\setminus \{k\}}| =|\mathrm{W}_{[T]\setminus \{k\}}\cup \mathrm{S}_k| = T}$, there exist $\alpha = T$ messages indexed by $\{1\}\cup{\mathrm{S}_{[T]\setminus \{k\}}}$ given which $\beta=T$ other messages indexed by ${\mathrm{W}_{[T]\setminus \{k\}}}\cup {\mathrm{S}_k}$ can be recovered.  
Next, suppose that $\mathrm{S}_{k,l}^{*}= \mathrm{W}_k$ for any $k,l\in [T]$. 
Fix arbitrary $k,l\in [T]$. 
Consider the two tuples $(\mathrm{R}^{*}_{k,l},\mathrm{S}^{*}_{k,l},\mathrm{W}^{*}_k)$ and $(\mathrm{R}^{*}_{l,k},\mathrm{S}^{*}_{l,k},\mathrm{W}^{*}_l)$, or equivalently, the two tuples $(\mathrm{W}_{l},\mathrm{S}^{*}_{k,l},\mathrm{S}_k)$ and $(\mathrm{W}_{k},\mathrm{S}^{*}_{l,k},\mathrm{S}_l)$. 
Note that ${\mathrm{W}_l\cup \mathrm{S}^{*}_{k,l} \rightarrow \mathrm{S}_k}$ and ${\mathrm{W}_k\cup \mathrm{S}^{*}_{l,k} \rightarrow \mathrm{S}_l}$. 
By assumption, ${\mathrm{S}^{*}_{k,l} = \mathrm{W}_k}$ and ${\mathrm{S}^{*}_{l,k} = \mathrm{W}_l}$. 
Thus, ${\mathrm{W}_l\cup \mathrm{W}_k \rightarrow \mathrm{S}_k}$ and ${\mathrm{W}_k\cup \mathrm{W}_l \rightarrow \mathrm{S}_l}$. 
This implies that ${\mathrm{W}_k\cup \mathrm{W}_l \rightarrow \mathrm{S}_k\cup \mathrm{S}_l}$. 
Since $\mathrm{W}_k\cup \mathrm{W}_l$ and $\mathrm{S}_k\cup \mathrm{S}_l$ are disjoint, and $|\mathrm{W}_k\cup \mathrm{W}_l| = |\mathrm{S}_k\cup \mathrm{S}_l|=2$, there exist $\alpha=2$ messages indexed by $\mathrm{W}_k\cup \mathrm{W}_l$ given which $\beta=2$ other messages indexed by $\mathrm{S}_k\cup \mathrm{S}_l$ can be recovered. 
\end{proof}

\begin{theorem}\label{thm:12}
For PIR-RSSI with $K$ messages, RSI's size ${M_1=1}$, and SSI's size ${M_2>1}$, the capacity is upper bounded by ${1/\min\{K-M_2-1,\lceil K/(M_2+1) \rceil\}}$.
\end{theorem}

\begin{proof}
We prove the lemma for ${M_2=2}$ for ease of exposition. 
The same proof technique can be used for any ${M_2>2}$. 
In this case, ${K> M_1+M_2=3}$. 
Note that for ${K=4}$ (or $K\geq 5$), ${\min\{K-3,\lceil K/3 \rceil\}}$ is given by $K-3$ (or ${\lceil K/3 \rceil}$).
Thus, for $K=4$ (or $K\geq 5$), we need to show that there exist $L\leq 3$ (or $L\leq \lfloor 2K/3 \rfloor$) messages given which all other $K-L$ messages can be recovered.  

\emph{$K=4$:} Taking $\mathrm{R}_1=\{1\}$ and $\mathrm{W}_1=\{4\}$, it is obvious that ${\mathrm{S}_1 = \{2,3\}}$, i.e., the tuple $(\{1\},\{2,3\},\{4\})$ exists. 
Thus, ${\{1,2,3\}\rightarrow \{4\}}$, i.e., given $L=3$ messages the other $K-L=1$ message can be recovered. 

\emph{$K\geq 5$:} 
By Lemma~\ref{lem:Grow}, it suffices to show that there exist ${\alpha\geq \lceil M_1M_2/(M_2+1)\rceil =1}$ messages given which ${\beta\geq \lceil \alpha/M_2\rceil=\lceil \alpha/2\rceil}$ other messages can be recovered. 

Let $\{(\mathrm{R}_i,\mathrm{S}_i,\mathrm{W}_i)\}_{i\in [T]}$ be a collection of tuples as defined earlier.
Note that $\mathrm{R}_i= \{1\}$ for all $i\in [T]$. 
There are two cases: 
(i) ${\mathrm{S}_j\cap (\mathrm{S}_i\cup \mathrm{W}_i)\neq \emptyset}$ for some ${1\leq i<j\leq T}$; and  
(ii) ${\mathrm{S}_j\cap (\mathrm{S}_i\cup \mathrm{W}_i)=\emptyset}$ for all ${1\leq i<j\leq T}$. 

First, consider the case (i). 
Let ${\mathrm{I} := \mathrm{S}_j \setminus (\mathrm{S}_i\cup \mathrm{W}_i)}$. 
Note that ${|\mathrm{I}|\leq 1}$. 
This is because ${|\mathrm{I}| = |\mathrm{S}_j| - |\mathrm{S}_j\cap (\mathrm{S}_i\cup \mathrm{W}_i)|}$, ${|\mathrm{S}_i|=2}$, and ${|\mathrm{S}_j\cap (\mathrm{S}_i\cup \mathrm{W}_i)|\geq 1}$ (by assumption). 
In this case, ${\{1\}\cup \mathrm{S}_i\cup \mathrm{I}\rightarrow \mathrm{W}_i\cup \mathrm{W}_j}$. 
This is because ${\{1\}\cup \mathrm{S}_i\rightarrow \mathrm{W}_i}$, ${\{1\}\cup \mathrm{S}_j\rightarrow \mathrm{W}_j}$, ${\mathrm{S}_j = \mathrm{I}\cup (\mathrm{S}_j\setminus \mathrm{I})}$, and ${(\mathrm{S}_j\setminus \mathrm{I}) \subset (\mathrm{S}_i\cup\mathrm{W}_i)}$.
Note that ${\{1\}\cup \mathrm{S}_i\cup \mathrm{I}}$ and ${\mathrm{W}_i\cup\mathrm{W}_j}$ are disjoint, and ${|\{1\}\cup \mathrm{S}_i\cup \mathrm{I}|= 3+|\mathrm{I}|}$ and ${|\mathrm{W}_i\cup\mathrm{W}_j|=2}$. 
Thus, there exist $\alpha=3+|\mathrm{I}|$ messages indexed by ${\{1\}\cup \mathrm{S}_i\cup \mathrm{I}}$ given which $\beta=2$ other messages indexed by ${\mathrm{W}_i\cup\mathrm{W}_j}$ can be recovered. 
Note that ${\alpha=3+|\mathrm{I}|\geq 1}$ and ${\beta=2\geq \lceil \alpha/2\rceil = \lceil (3+|\mathrm{I}|)/2\rceil}$ since ${|\mathrm{I}|\leq 1}$. 

Now, consider the case (ii). 
First, suppose that for some $k,l\in [T]$, some ${\mathrm{R}_{k,l}^{*}\subset \mathrm{S}_l}$, and ${\mathrm{W}_k^{*}=\mathrm{W}_k}$, there exists a tuple ${(\mathrm{R}_{k,l}^{*},\mathrm{S}_{k,l}^{*},\mathrm{W}_k^{*})}$ such that  ${\mathrm{S}_{k,l}^{*}\neq \mathrm{S}_k}$.
Let ${\mathrm{I} := \mathrm{S}^{*}_{k,l}\cap \mathrm{S}_k}$. 
Note that ${|\mathrm{I}|\leq 1}$. 
This is because ${\mathrm{S}^{*}_{k,l}\neq \mathrm{S}_k}$ and ${|\mathrm{S}^{*}_{k,l}|=|\mathrm{S}_k|=2}$. 
It is easy to verify that ${\{1\}\cup \mathrm{S}_{[T]\setminus \{k\}}\cup \mathrm{I}\rightarrow \mathrm{W}_{[T]\setminus \{k\}}\cup \mathrm{W}_k}$.
Note that ${\{1\}\cup \mathrm{S}_{[T]\setminus \{k\}}\cup \mathrm{I}}$ and ${\mathrm{W}_{[T]\setminus \{k\}}\cup\mathrm{W}_k}$ are disjoint. 
Moreover, ${|\{1\}\cup \mathrm{S}_{[T]\setminus \{k\}}\cup \mathrm{I}|= 2T+|\mathrm{I}|-1}$ and ${|\mathrm{W}_{[T]\setminus \{k\}}\cup\mathrm{W}_k|= T}$. 
This implies that there exist ${\alpha=2T+|\mathrm{I}|-1}$ messages indexed by ${\{1\}\cup \mathrm{S}_{[T]\setminus \{k\}}\cup \mathrm{I}}$ given which ${\beta=T}$ other messages indexed by ${\mathrm{W}_{[T]\setminus \{k\}}\cup\mathrm{W}_k}$ can be recovered. 
Note that ${\alpha=2T+|\mathrm{I}|-1\geq 1}$ since ${T\geq 1}$, and ${\beta=T\geq \lceil \alpha/2\rceil =T+\lceil (|\mathrm{I}|-1)/2\rceil}$ since ${|\mathrm{I}|\leq 1}$. 
Next, suppose that for any ${k,l\in [T]}$, any ${\mathrm{R}_{k,l}^{*}\subset \mathrm{S}_l}$, and ${\mathrm{W}_k^{*}=\mathrm{W}_k}$, there exists a tuple ${(\mathrm{R}_{k,l}^{*},\mathrm{S}_{k,l}^{*},\mathrm{W}_k^{*})}$ such that ${\mathrm{S}_{k,l}^{*}= \mathrm{S}_k}$. 
Fix arbitrary $k,l\in [T]$. 
Consider the two tuples ${(\mathrm{R}^{*}_{k,l},\mathrm{S}^{*}_{k,l},\mathrm{W}^{*}_k)}$ and ${(\mathrm{R}^{*}_{l,k},\mathrm{S}^{*}_{l,k},\mathrm{W}^{*}_l)}$, or equivalently, the two tuples ${(\mathrm{R}^{*}_{k,l},\mathrm{S}_{k},\mathrm{W}_k)}$ and ${(\mathrm{R}^{*}_{l,k},\mathrm{S}_{l},\mathrm{W}_l)}$. 
Note that ${\mathrm{R}^{*}_{k,l}\cup \mathrm{S}_k\rightarrow \mathrm{W}_k}$ and ${\mathrm{R}^{*}_{l,k}\cup \mathrm{S}_l\rightarrow \mathrm{W}_l}$. 
Since ${\mathrm{R}^{*}_{k,l}\subset \mathrm{S}_l}$ and ${\mathrm{R}^{*}_{l,k}\subset \mathrm{S}_k}$ (by assumption), then ${\mathrm{S}_k\cup \mathrm{S}_l \rightarrow \mathrm{W}_k\cup \mathrm{W}_l}$. 
Since ${\mathrm{S}_k\cup \mathrm{S}_l}$ and ${\mathrm{W}_k\cup \mathrm{W}_l}$ are disjoint, and ${|\mathrm{S}_k\cup \mathrm{S}_l| = 4}$ and ${|\mathrm{W}_k\cup \mathrm{W}_l|=2}$, there exist ${\alpha=4}$ (${\geq 1}$) messages indexed by ${\mathrm{S}_k\cup \mathrm{S}_l}$ given which ${\beta=2}$ (${\geq \lceil \alpha/2\rceil = 2}$) other messages indexed by ${\mathrm{W}_k\cup \mathrm{W}_l}$ can be recovered. 
\end{proof}

\begin{theorem}\label{thm:21}
For PIR-RSSI with $K$ messages, RSI's size ${M_1>1}$, and SSI's size ${M_2=1}$, the capacity is upper bounded by ${1/\min\{K-M_1-1,\lceil K/2 \rceil\}}$. 
\end{theorem}

\begin{proof}
For ease of exposition, we prove the lemma for ${M_1=2}$. 
The proof for any ${M_1>2}$ follows from similar arguments. 
In this case, ${K> M_1+M_2=3}$. 
Note that for ${K\in \{4,5\}}$ (or $K\geq 6$), ${\min\{K-3,\lceil K/2 \rceil\}}$ is given by $K-3$ (or ${\lceil K/2 \rceil}$).
Thus, for $K\in \{4,5\}$ (or $K\geq 6$), we need to show that there exist $L\leq 3$ (or $L\leq \lfloor K/2 \rfloor$) messages given which all other $K-L$ messages can be recovered.  

\emph{$K\in \{4,5\}$:} For $K=4$, the tuple $(\{1,2\},\{3\},\{4\})$ exists. 
Thus, given $L=3$ messages indexed by $\{1,2,3\}$ the other $K-L=1$ message indexed by $\{4\}$ can be recovered. 
For $K=5$, without loss of generality, we can assume that the tuples $(\{1,2\},\{3\},\{4\})$ and $(\{1,2\},\{3\},\{5\})$ exist. 
Thus, given $L=3$ messages indexed by $\{1,2,3\}$ the other $K-L=2$ messages indexed by $\{4,5\}$ can be recovered. 

\emph{$K\geq 6$:} 
By Lemma~\ref{lem:Grow}, it suffices to show that there exist ${\alpha\geq \lceil M_1M_2/(M_2+1)\rceil =1}$ messages given which ${\beta\geq \lceil \alpha/M_2\rceil=\alpha}$ other messages can be recovered. 

Let ${\{(\mathrm{R}_i,\mathrm{S}_i,\mathrm{W}_i)\}_{i\in [T]}}$ be a collection of tuples as defined earlier. 
Note that ${\mathrm{R}_i = \{1,2\}}$ for all ${i\in [T]}$. 
There are four cases: 
(i) there exist some ${1\leq i<j<k\leq T}$ such that ${\mathrm{S}_j\subset (\mathrm{S}_i\cup \mathrm{W}_i)}$ and ${\mathrm{S}_k\subset ((\mathrm{S}_i\cup \mathrm{W}_i) \cup (\mathrm{S}_j\cup \mathrm{W}_j))}$; 
(ii) there exist some ${1\leq i<j<k<l\leq T}$ such that ${\mathrm{S}_j\subset (\mathrm{S}_i\cup \mathrm{W}_i)}$, ${\mathrm{S}_k\not\subset ((\mathrm{S}_i\cup \mathrm{W}_i) \cup (\mathrm{S}_j\cup \mathrm{W}_j))}$, and ${\mathrm{S}_l\subset (\mathrm{S}_k\cup \mathrm{W}_k)}$; 
(iii) there exist unique ${1\leq i<j\leq T}$ such that ${\mathrm{S}_j\subset (\mathrm{S}_i\cup\mathrm{W}_i)}$; and
(iv) there do not exist any ${1\leq i<j\leq T}$ such that ${\mathrm{S}_j\subset (\mathrm{S}_i\cup\mathrm{W}_i)}$.
Note that the cases (i) and (ii) are not mutually exclusive. 

First, consider the case (i).
In this case, it is easy to show that ${\{1,2\}\cup \mathrm{S}_i\rightarrow \mathrm{W}_i\cup\mathrm{W}_j\cup\mathrm{W}_k}$.
This is because ${\{1,2\}\cup \mathrm{S}_i\rightarrow \mathrm{W}_i}$, ${\{1,2\}\cup \mathrm{S}_j\rightarrow \mathrm{W}_j}$, ${\{1,2\}\cup \mathrm{S}_k\rightarrow \mathrm{W}_k}$, ${\mathrm{S}_j\subset (\mathrm{S}_i\cup \mathrm{W}_i)}$, and ${\mathrm{S}_k\subset ((\mathrm{S}_i\cup \mathrm{W}_i) \cup (\mathrm{S}_j\cup \mathrm{W}_j))}$.
Since ${\{1,2\}\cup\mathrm{S}_i}$ and ${\mathrm{W}_i\cup\mathrm{W}_j\cup\mathrm{W}_k}$ are disjoint, and ${|\{1,2\}\cup\mathrm{S}_i|=3}$ and ${|\mathrm{W}_i\cup\mathrm{W}_j\cup\mathrm{W}_k|=3}$, there exist $\alpha=3$ messages given which $\beta=3$ ($\geq \alpha$) other messages can be recovered. 

Now, consider the case (ii). 
In this case, it can be shown that ${\{1,2\}\cup \mathrm{S}_i\cup \mathrm{S}_k\rightarrow \mathrm{W}_i\cup\mathrm{W}_j\cup\mathrm{W}_k\cup \mathrm{W}_l}$. 
Note that ${\{1,2\}\cup\mathrm{S}_i\cup \mathrm{S}_k}$ and ${\mathrm{W}_i\cup\mathrm{W}_j\cup\mathrm{W}_k\cup \mathrm{W}_l}$ are disjoint, and ${|\{1,2\}\cup\mathrm{S}_i\cup \mathrm{S}_k|=4}$ and ${|\mathrm{W}_i\cup\mathrm{W}_j\cup\mathrm{W}_k\cup \mathrm{W}_l|=4}$. 
This implies that there exist ${\alpha=4}$ messages given which ${\beta=4}$ other messages can be recovered.

Next, consider the case (iii). 
Recall that in this case there exist unique ${1\leq i<j\leq T}$ such that ${\mathrm{S}_j\subset \mathrm{S}_i\cup \mathrm{W}_i}$. 
It is easy to verify that this case can only occur for odd $K$, i.e., ${K\in \{7,9,11,\dots\}}$. 
In addition, one can easily see that ${T = (K-1)/2\geq 3}$.
Note that ${[T]\setminus \{i,j\}\neq \emptyset}$. 
By assumption, for any ${k\in [T]\setminus \{i,j\}}$, ${\mathrm{S}_k\cup \mathrm{W}_k}$ and ${(\mathrm{S}_i\cup \mathrm{W}_i)\cup (\mathrm{S}_j\cup \mathrm{W}_j)}$ are disjoint. 

Fix an arbitrary ${k\in [T]\setminus \{i,j\}}$. 
Let $\mathrm{R}_k^{*}$ be an arbitrary $2$-subset of ${[K]\setminus (\mathrm{S}_k \cup \mathrm{W}_k)}$. 
By Lemma~\ref{lem:NC}, there exists a tuple ${(\mathrm{R}_k^{*},\mathrm{S}_k^{*},\mathrm{W}_k)}$ for some $\mathrm{S}^{*}_k$. 
That is, ${\mathrm{R}_k^{*}\cup \mathrm{S}^{*}_k\rightarrow \mathrm{W}_k}$.

First, suppose that ${\mathrm{S}^{*}_k\neq \mathrm{S}_k}$. 
Thus, ${\mathrm{S}^{*}_k\subset [K]\setminus (\mathrm{S}_k \cup \mathrm{W}_k)}$.
It is easy to see that ${\{1,2\}\cup \mathrm{S}_i\rightarrow \mathrm{W}_i\cup \mathrm{W}_j}$.
This is because ${\{1,2\}\cup \mathrm{S}_i\rightarrow \mathrm{W}_i}$, 
${\{1,2\}\cup \mathrm{S}_j\rightarrow \mathrm{W}_j}$, and 
${\mathrm{S}_j\subset \mathrm{S}_i\cup \mathrm{W}_i}$.
Note also that ${\{1,2\}\cup \mathrm{S}_{[T]\setminus \{i,j,k\}}\rightarrow \mathrm{W}_{[T]\setminus \{i,j,k\}}}$. 
This is because ${\{1,2\}\cup \mathrm{S}_l\rightarrow \mathrm{W}_l}$ for all ${l\in [T]\setminus \{i,j,k\}}$.
Thus, ${\{1,2\}\cup \mathrm{S}_{[T]\setminus \{j,k\}}\rightarrow \mathrm{W}_{[T]\setminus \{k\}}}$. 
By assumption, ${(\mathrm{R}_k^{*}\cup \mathrm{S}^{*}_k)\subset \{1,2\}\cup  \mathrm{S}_{[T]\setminus \{k\}}\cup \mathrm{W}_{[T]\setminus \{k\}}}$. 
In addition, ${\mathrm{R}_k^{*}\cup \mathrm{S}^{*}_k\rightarrow \mathrm{W}_k}$. 
By combining these results, it then follows that ${\{1,2\}\cup \mathrm{S}_{[T]\setminus \{j,k\}}\rightarrow \mathrm{W}_{[T]}}$.
It is easy to verify that 
${\{1,2\}\cup \mathrm{S}_{[T]\setminus \{j,k\}}}$ and ${\mathrm{W}_{[T]}}$ are disjoint, and 
${|\{1,2\}\cup \mathrm{S}_{[T]\setminus \{j,k\}}|=|\mathrm{W}_{[T]}| = T}$. 
This implies that there exist ${\alpha=T}$ messages given which ${\beta=T}$ other messages can be recovered. 

Next, suppose that $\mathrm{S}^{*}_l = \mathrm{S}_l$ for all ${l\in [T]\setminus \{i,j\}}$.
That is, for any ${l\in [T]\setminus \{i,j\}}$, ${\mathrm{R}^{*}_l\cup \mathrm{S}_l\rightarrow \mathrm{W}_l}$ for any arbitrary $2$-subset ${\mathrm{R}^{*}_l\subset [K]\setminus (\mathrm{S}_l\cup \mathrm{W}_l)}$.
Recall that ${\mathrm{S}_j\subset \mathrm{S}_i\cup \mathrm{W}_i}$.
Thus, 
${\mathrm{S}_j = \mathrm{S}_i}$, or 
${\mathrm{S}_j = \mathrm{W}_i}$. 
We only present the proof for the case of ${\mathrm{S}_j = \mathrm{S}_i}$. 
The proof for the case of ${\mathrm{S}_j = \mathrm{W}_i}$ is the same except that $\mathrm{S}_i$ needs to be replaced by $\mathrm{W}_i$ everywhere.  

Fix an arbitrary ${k\in [T]\setminus \{i,j\}}$. 
By Lemma~\ref{lem:NC}, there exists a tuple 
${(\mathrm{S}_k\cup \mathrm{W}_k,\mathrm{S}^{*},\mathrm{W}_i)}$ for some $\mathrm{S}^{*}$. 
That is, ${\mathrm{S}_{k}\cup\mathrm{W}_k\cup \mathrm{S}^{*}\rightarrow \mathrm{W}_i}$.
There are four cases: 
(iii-1) ${\mathrm{S}^{*}=\mathrm{S}_i}$; 
(iii-2) ${\mathrm{S}^{*}=\mathrm{W}_j}$;
(iii-3) ${\mathrm{S}^{*}\subset \{1,2\}}$; and
(iii-4) ${\mathrm{S}^{*}\subset \mathrm{S}_{[T]\setminus \{i,j,k\}}\cup \mathrm{W}_{[T]\setminus \{i,j,k\}}}$. 

First, consider the case (iii-1). 
In this case, ${\mathrm{S}^{*}=\mathrm{S}_i}$.
That is, ${\mathrm{S}_{k}\cup\mathrm{W}_k\cup \mathrm{S}_{i}\rightarrow \mathrm{W}_i}$.
Lemma~\ref{lem:NC} implies that there exists a tuple ${(\mathrm{S}_{k}\cup\mathrm{W}_k,\mathrm{S}^{*}_j,\mathrm{W}_j)}$ for some $\mathrm{S}^{*}_j$.
That is, ${\mathrm{S}_{k}\cup\mathrm{W}_k\cup \mathrm{S}^{*}_j\rightarrow \mathrm{W}_j}$. 
There are two cases: ${\mathrm{S}^{*}_j\neq \{2\}}$, and ${\mathrm{S}^{*}_j= \{2\}}$. 
We present the proof for ${\mathrm{S}^{*}_j\neq \{2\}}$. 
The proof for $\mathrm{S}^{*}_j= \{2\}$ is the same except that $\{1\}$ needs to be replaced by $\{2\}$ everywhere.
Assume that ${\mathrm{S}^{*}_j\neq \{2\}}$. 
Take ${\mathrm{R}^{*}_k = \{1\}\cup \mathrm{S}_i}$. 
Since ${k\in [T]\setminus \{i,j\}}$ and ${\mathrm{R}^{*}_k\subset [K]\setminus (\mathrm{S}_k\cup \mathrm{W}_k)}$, then ${\mathrm{R}^{*}_k\cup\mathrm{S}_k\rightarrow \mathrm{W}_k}$ (by assumption). 
Equivalently, ${\{1\}\cup \mathrm{S}_i\cup \mathrm{S}_k\rightarrow \mathrm{W}_k}$. 
Similarly, it can be shown that ${\{1\}\cup \mathrm{S}_i\cup \mathrm{S}_l\rightarrow \mathrm{W}_l}$ for all ${l\in [T]\setminus \{i,j\}}$. 
Thus, ${\{1\}\cup \mathrm{S}_i\cup  \mathrm{S}_{[T]\setminus \{i,j\}})\rightarrow \mathrm{W}_{[T]\setminus \{i,j\}}}$, or equivalently, 
${\{1\}\cup \mathrm{S}_{[T]\setminus \{j\}}\rightarrow \mathrm{W}_{[T]\setminus \{i,j\}}}$. 
Recall that ${\mathrm{S}_{k}\cup\mathrm{W}_k\cup \mathrm{S}^{*}\rightarrow \mathrm{W}_{i}}$. 
Since ${\mathrm{S}^{*}=\mathrm{S}_i}$ by assumption, 
then ${\mathrm{S}_{k}\cup\mathrm{W}_k\cup \mathrm{S}_{i}\rightarrow \mathrm{W}_{i}}$.  
By combining ${\{1\}\cup \mathrm{S}_{[T]\setminus \{j\}}\rightarrow \mathrm{W}_{[T]\setminus \{i,j\}}}$ and ${\mathrm{S}_k\cup \mathrm{W}_k\cup \mathrm{S}_i\rightarrow \mathrm{W}_i}$, 
it readily follows that ${\{1\}\cup  \mathrm{S}_{[T]\setminus \{j\}}\rightarrow  \mathrm{W}_{[T]\setminus \{j\}}}$.
Recall also that $\mathrm{S}_k\cup \mathrm{W}_k\cup \mathrm{S}^{*}_j\rightarrow \mathrm{W}_j$.
Since ${\mathrm{S}^{*}_j\neq \{2\}}$ and ${\mathrm{S}^{*}_i\neq \mathrm{S}_i=\mathrm{S}_j}$, then ${\mathrm{S}^{*}_j\subset \{1\}\cup  \mathrm{S}_{[T]}\cup\mathrm{W}_{[T]}}$. 
By combining ${\{1\}\cup \mathrm{S}_{[T]\setminus \{j\}}\rightarrow  \mathrm{W}_{[T]\setminus \{j\}}}$ and $\mathrm{S}_k\cup \mathrm{W}_k\cup \mathrm{S}^{*}_j\rightarrow \mathrm{W}_j$, 
it then follows that ${\{1\}\cup \mathrm{S}_{[T]\setminus \{j\}}\rightarrow \mathrm{W}_{[T]}}$. 
It is easy to verify that ${\{1\}\cup  \mathrm{S}_{[T]\setminus \{j\}}}$ and ${\mathrm{W}_{[T]}}$ are disjoint, and ${|\{1\}\cup \mathrm{S}_{[T]\setminus \{j\}}|=|\mathrm{W}_{[T]}|=T}$. 
Thus, there exist ${\alpha=T}$ messages given which ${\beta=T}$ other messages can be recovered.


Next, consider the case (iii-2). 
In this case, ${\mathrm{S}^{*}=\mathrm{W}_j}$.
That is, ${\mathrm{S}_{k}\cup\mathrm{W}_k\cup \mathrm{W}_{j}\rightarrow \mathrm{W}_i}$.
By Lemma~\ref{lem:NC}, there exists a tuple ${(\mathrm{S}_{k}\cup\mathrm{W}_k,\mathrm{S}^{*}_i,\mathrm{S}_i)}$ for some $\mathrm{S}^{*}_i$.
That is, ${\mathrm{S}_{k}\cup\mathrm{W}_k\cup \mathrm{S}^{*}_i\rightarrow \mathrm{S}_i}$. 
There are two cases: ${\mathrm{S}^{*}_i\neq \{2\}}$, and ${\mathrm{S}^{*}_i= \{2\}}$. 
Similarly as in the case (iii-1), we present the proof for ${\mathrm{S}^{*}_i\neq \{2\}}$. 
The proof for $\mathrm{S}^{*}_i= \{2\}$ is the same except that $\{1\}$ needs to be replaced by $\{2\}$ everywhere.
Assume that ${\mathrm{S}^{*}_i\neq \{2\}}$. 
Using similar arguments as in the case (iii-1) and noting that ${\mathrm{S}_{k}\cup\mathrm{W}_k\cup \mathrm{W}_{j}\rightarrow \mathrm{W}_{i}}$, it can be shown that
${\{1\}\cup \mathrm{S}_{[T]\setminus \{i,j\}}\cup \mathrm{W}_j\rightarrow  \mathrm{W}_{[T]\setminus \{j\}}}$.
Recall that $\mathrm{S}_k\cup \mathrm{W}_k\cup \mathrm{S}^{*}_i\rightarrow \mathrm{S}_i$.
Since ${\mathrm{S}^{*}_i\neq \{2\}}$ and ${\mathrm{S}^{*}_i\neq \mathrm{S}_i=\mathrm{S}_j}$, 
then ${\mathrm{S}^{*}_i\subset \{1\}\cup \mathrm{S}_{[T]\setminus \{j\}}\cup\mathrm{W}_{[T]\setminus \{j\}}\cup \mathrm{W}_j}$. 
Thus, by combining ${\{1\}\cup \mathrm{S}_{[T]\setminus \{i,j\}})\cup \mathrm{W}_j\rightarrow \mathrm{W}_{[T]\setminus \{j\}}}$ and $\mathrm{S}_k\cup \mathrm{W}_k\cup \mathrm{S}^{*}_i\rightarrow \mathrm{S}_i$, 
it follows that ${\{1\}\cup \mathrm{S}_{[T]\setminus \{i,j\}}\cup \mathrm{W}_j\rightarrow \mathrm{W}_{[T]\setminus \{j\}}\cup \mathrm{S}_i}$. 
One can easily verify that ${\{1\}\cup  \mathrm{S}_{[T]\setminus \{i,j\}}\cup \mathrm{W}_j}$ and ${\mathrm{W}_{[T]\setminus \{j\}}\cup \mathrm{S}_i}$ are disjoint, and ${|\{1\}\cup \mathrm{S}_{[T]\setminus \{i,j\}}\cup \mathrm{W}_j|=|\mathrm{W}_{[T]\setminus \{j\}}\cup \mathrm{S}_i|=T}$. 
Thus, there exist ${\alpha=T}$ messages given which ${\beta=T}$ other messages can be recovered.

Now, consider the case (iii-3). 
In this case, ${\mathrm{S}^{*}\subset \{1,2\}}$. 
Without loss of generality, assume that ${\mathrm{S}^{*} = \{1\}}$. 
That is, ${\mathrm{S}_{k}\cup\mathrm{W}_k\cup \{1\}\rightarrow \mathrm{W}_i}$.
Lemma~\ref{lem:NC} implies that there exists a tuple ${(\mathrm{S}_{k}\cup\mathrm{W}_k,\mathrm{S}^{*}_i,\mathrm{S}_i)}$ for some $\mathrm{S}^{*}_i$.
That is, ${\mathrm{S}_{k}\cup\mathrm{W}_k\cup \mathrm{S}^{*}_i\rightarrow \mathrm{S}_i}$. 
There are two cases: ${\mathrm{S}^{*}_i\neq \{2\}}$, and ${\mathrm{S}^{*}_i= \{2\}}$.
First, suppose that ${\mathrm{S}^{*}_i\neq \{2\}}$. 
Using similar arguments as in the case (iii-2) and noting that ${\mathrm{S}_{k}\cup\mathrm{W}_k\cup \{1\}\rightarrow \mathrm{W}_i}$ and ${\mathrm{S}_{k}\cup\mathrm{W}_k\cup \mathrm{S}^{*}_i\rightarrow \mathrm{S}_i}$, one can show that 
${\{1\}\cup \mathrm{S}_{[T]\setminus \{i,j\}}\cup \mathrm{W}_j\rightarrow  \mathrm{W}_{[T]\setminus \{j\}}\cup \mathrm{S}_i}$.
Similarly as in the case (iii-2), it then follows that there exist ${\alpha=T}$ messages given which ${\beta=T}$ (${\geq \alpha}$) other messages can be recovered.
Now, suppose that $\mathrm{S}^{*}_i = \{2\}$. 
Using similar arguments as before, 
it can be shown that 
${\{1,2\}\cup \mathrm{S}_{[T]\setminus \{i,j\}}\rightarrow \mathrm{W}_{[T]\setminus \{j\}}\cup \mathrm{S}_i}$. 
It is easy to verify that ${\{1,2\}\cup  \mathrm{S}_{[T]\setminus \{i,j\}}}$ and ${ \mathrm{W}_{[T]\setminus \{j\}}\cup \mathrm{S}_i}$ are disjoint, and ${|\{1,2\}\cup \mathrm{S}_{[T]\setminus \{i,j\}}|=|\mathrm{W}_{[T]\setminus \{j\}}\cup \mathrm{S}_i|=T}$. 
Thus, there exist ${\alpha=T}$ messages given which ${\beta=T}$ other messages can be recovered.

Lastly, consider the case (iii-4). 
In this case, ${\mathrm{S}^{*}\subset  \mathrm{S}_{[T]\setminus \{i,j,k\}}\cup \mathrm{W}_{[T]\setminus \{i,j,k\}}}$.
Using similar arguments as in the previous cases, it can be shown that ${\{1\}\cup  \mathrm{S}_{[T]\setminus \{i,j\}}\rightarrow \mathrm{W}_{[T]\setminus \{i,j\}}}$.
Note that ${\mathrm{S}_{k}\cup\mathrm{W}_k\cup \mathrm{S}^{*}\rightarrow \mathrm{W}_i}$ and 
${\mathrm{S}^{*}\subset  \mathrm{S}_{[T]\setminus \{i,j\}}\cup \mathrm{W}_{[T]\setminus \{i,j\}}}$. 
Thus,
${\{1\}\cup \mathrm{S}_{[T]\setminus \{i,j\}}\rightarrow \mathrm{W}_{[T]\setminus \{j\}}}$. 
One can easily verify that ${\{1\}\cup  \mathrm{S}_{[T]\setminus \{i,j\}}}$ and ${ \mathrm{W}_{[T]\setminus \{j\}}}$ are disjoint, and ${|\{1\}\cup  \mathrm{S}_{[T]\setminus \{i,j\}}|=|\mathrm{W}_{[T]\setminus \{j\}}|=T-1}$. 
Thus, there exist ${\alpha=T-1}$ messages given which ${\beta=T-1}$ other messages can be recovered.

Now, consider the case (iv). 
Recall that in this case there do not exist any ${1\leq i<j\leq T}$ such that ${\mathrm{S}_j\subset (\mathrm{S}_i\cup\mathrm{W}_i)}$.
It is easy to verify that this case can only occur for even $K$, i.e., ${K\in \{6,8,10,\dots\}}$. 
Note also that ${T = K/2-1\geq 2}$.

Fix an arbitrary ${i\in [T]}$. 
Let $\mathrm{R}_i^{*}$ be an arbitrary $2$-subset of ${[K]\setminus (\mathrm{S}_i \cup \mathrm{W}_i)}$. 
By Lemma~\ref{lem:NC}, there exists a tuple ${(\mathrm{R}_i^{*},\mathrm{S}_i^{*},\mathrm{S}_i)}$ for some $\mathrm{S}^{*}_i$. 
That is, ${\mathrm{R}_i^{*}\cup \mathrm{S}^{*}_i\rightarrow \mathrm{S}_i}$.

First, suppose that ${\mathrm{S}^{*}_i\neq \mathrm{W}_i}$. 
Thus, ${\mathrm{S}^{*}_i\subset [K]\setminus (\mathrm{S}_i \cup \mathrm{W}_i)}$.
Note that ${\{1,2\}\cup \mathrm{S}_{[T]\setminus \{i\}}\rightarrow  \mathrm{W}_{[T]\setminus \{i\}}}$. 
This is because ${\{1,2\}\cup \mathrm{S}_l\rightarrow \mathrm{W}_l}$ for all ${l\in [T]\setminus \{i\}}$.
By assumption, ${(\mathrm{R}_i^{*}\cup \mathrm{S}^{*}_i)\subset \{1,2\}\cup  \mathrm{S}_{[T]\setminus \{i\}}\cup \mathrm{W}_{[T]\setminus \{i\}}}$. 
Moreover, ${\mathrm{R}_i^{*}\cup \mathrm{S}^{*}_i\rightarrow \mathrm{S}_i}$. 
By combining these results, it follows that ${\{1,2\}\cup \mathrm{S}_{[T]\setminus \{i\}}\rightarrow  \mathrm{W}_{[T]\setminus \{i\}}\cup \mathrm{S}_i}$.
Note that 
${\{1,2\}\cup \mathrm{S}_{[T]\setminus \{i\}}\cup \mathrm{W}_{[T]\setminus \{i\}}\cup \mathrm{S}_i = [K]\setminus \mathrm{W}_i}$.
Recall also that ${\{1,2\}\cup \mathrm{S}_{i}\rightarrow \mathrm{W}_i}$. 
Obviously, ${\mathrm{S}^{*}\subset [K]\setminus \mathrm{W}_i}$. 
This immediately implies that 
${\{1,2\}\cup \mathrm{S}_{[T]\setminus \{i\}}\rightarrow \mathrm{W}_{[T]}\cup \mathrm{S}_i}$. 
Note that
${\{1,2\}\cup \mathrm{S}_{[T]\setminus \{i\}}}$ and ${\mathrm{W}_{[T]}\cup \mathrm{S}_i}$ are disjoint, and 
${|\{1,2\}\cup \mathrm{S}_{[T]\setminus \{i\}}|=|\mathrm{W}_{[T]})\cup \mathrm{S}_i| = T+1}$. 
Thus, there exist ${\alpha=T+1}$ messages given which ${\beta=T+1}$ other messages can be recovered. 

Now, suppose that $\mathrm{S}^{*}_l = \mathrm{W}_l$ for all ${l\in [T]}$.
That is, for any ${l\in [T]}$, ${\mathrm{R}^{*}_l\cup \mathrm{W}_l\rightarrow \mathrm{S}_l}$ for any arbitrary $2$-subset ${\mathrm{R}^{*}_l\subset [K]\setminus (\mathrm{S}_l\cup \mathrm{W}_l)}$.
Let ${\mathrm{R}^{*}_l = \{1\} \cup \mathrm{W}_{l+1}}$ for all ${1\leq l<T}$, and let ${\mathrm{R}^{*}_T = \{1\} \cup \mathrm{W}_{1}}$. 
Observe that ${\{1\}\cup \mathrm{W}_l\rightarrow \mathrm{S}_l}$ for all ${1\leq l<T}$, and ${\{1\}\cup \mathrm{W}_1\rightarrow \mathrm{S}_T}$. 
Thus, ${\{1\}\cup \mathrm{W}_{[T]}\rightarrow \mathrm{S}_{[T]}}$. 
Note that ${\{1\}\cup \mathrm{W}_{[T]}\cup\mathrm{S}_{[T]} = [K]\setminus \{2\}}$. 
By Lemma~\ref{lem:NC}, for any arbitrary $2$-subset ${\mathrm{R}^{*}\subset [K]\setminus \{2\}}$, there exists a tuple ${(\mathrm{R}^{*},\mathrm{S}^{*},\{2\})}$ for some ${\mathrm{S}^{*}\subset [K]\setminus (\{2\}\cup \mathrm{R}^{*})}$.
That is, ${\mathrm{R}^{*}\cup \mathrm{S}^{*}\rightarrow \{2\}}$. 
Since ${\mathrm{R}^{*}\cup \mathrm{S}^{*}\subset [K]\setminus \{2\}}$, then ${\{1\}\cup  \mathrm{W}_{[T]}\rightarrow \mathrm{S}_{[T]}\cup \{2\}}$. 
Note that ${\{1\}\cup \mathrm{W}_{[T]}}$ and ${\mathrm{S}_{[T]}\cup \{2\}}$ are disjoint, and ${|\{1\}\cup \mathrm{W}_{[T]})| = |\mathrm{S}_{[T]}\cup \{2\}| = T+1}$. 
Thus, there exist ${\alpha=T+1}$ messages given which ${\beta=T+1}$ other messages can be recovered. 
\end{proof}

\section{Proof of Theorem~\ref{thm:Ach}}\label{sec:Ach}
In this section, we present the achievability scheme. 
To achieve the rate $1/(K-M_1-M_2)$ or $1/\lceil K/(M_2+1)\rceil$, we use a modified version of the MDS Code scheme of~\cite{KGHERS2020} or a modified version of the Partition-and-Code scheme of~\cite{KGHERS2020}, respectively. 
Note that the MDS Code scheme was originally designed for ${M_1\geq 1, M_2=0}$, and the Partition-and-Code scheme was originally designed for ${M_1=0,M_2\geq 1}$. 

\emph{MDS Code based PIR-RSSI Scheme:} 
For ease of notation, we define ${P := K-M_1-M_2}$. 
First, the user constructs an arbitrary $P\times K$ matrix $\mathrm{G}$ that generates a ${[K,P]}$ MDS code (over $\mathbbmss{F}_q$).
The user then sends $\mathrm{G}$ to the server as the query. 
For each ${i\in [P]}$, let $\mathrm{g}_i$ be the $i$th row of $\mathrm{G}$.  
Given $\mathrm{G}$, the server computes ${\mathrm{Y}_i=\mathrm{g}_i[\mathrm{X}^{\mathsf{T}}_1,\dots,\mathrm{X}^{\mathsf{T}}_K]^{\mathsf{T}}}$ for each ${i\in [P]}$, and sends ${\mathrm{Y}_1,\dots,\mathrm{Y}_{P}}$ back to the user as the answer. 

The rate of the MDS Code based scheme is equal to ${1/P}$, or equivalently, ${1/(K-M_1-M_2)}$. 
This is because ${\mathrm{g}_1,\dots,\mathrm{g}_{P}}$ are linearly independent, and ${\mathbf{X}_1,\dots,\mathbf{X}_K}$ are independent and uniformly distributed over $\mathbbmss{F}_q^n$, and consequently, ${\mathbf{Y}_1,\dots,\mathbf{Y}_{P}}$ are independent and uniformly distributed over $\mathbbmss{F}_q^n$. 
Thus, the amount of information downloaded from the server is ${PB}$ bits, where ${B=n\log_2 q}$ is the amount of information in a message (in bits). 
Since the amount of information required by the user is $B$ bits, the rate of the scheme is equal to ${B/(PB) = 1/P}$. 

By the properties of MDS codes, it is easy to show that 
the MDS Code based scheme satisfies the privacy condition. 
Note that the minimum distance of the MDS code generated by the matrix $\mathrm{G}$ is equal to ${K-P+1 = M_1+M_2+1}$. 
Thus, for any $M_1$-subset ${\mathrm{R}^{*}\subset [K]}$, any $M_2$-subset ${\mathrm{S}^{*}\subset [K]\setminus \mathrm{R}^{*}}$, and any ${\mathrm{W}^{*}\in [K]\setminus (\mathrm{R}^{*}\cup\mathrm{S}^{*})}$, there exists a row-vector---unique up to scalar multiplication---in the row space of $\mathrm{G}$ whose support is $\mathrm{R}^{*}\cup\mathrm{S}^{*}\cup\{\mathrm{W}^{*}\}$.
This implies that given $\mathrm{G}$, the probability that the RSI's index set is $\mathrm{R}^{*}$, the SSI's index set is $\mathrm{S}^{*}$, and the demand index is $\mathrm{W}^{*}$, is the same for every $(\mathrm{R}^{*},\mathrm{S}^{*},\mathrm{W}^{*})$ defined as above. 
Thus, the MDS Code based scheme not only satisfies the privacy condition by preserving the privacy of the demand index and the RSI's index set, but it also keeps the SSI's index set private. 
It is also easy to see that the MDS Code based scheme satisfies the recoverability condition. 
Since ${|\mathrm{R}\cup\mathrm{S}\cup\mathrm{W}|=M_1+M_2+1}$, there exists a row-vector---unique up to scalar multiplication---in the row space of $\mathrm{G}$ whose support is $\mathrm{R}\cup\mathrm{S}\cup\mathrm{W}$.
This implies that there exists ${\mathrm{Y}\in \mathbbmss{F}_q^n}$ in the linear span of ${\mathrm{Y}_1,\dots,\mathrm{Y}_{P}}$ such that $\mathrm{Y}$ is a linear combination of the demand message $\mathrm{X}_{\mathrm{W}}$, the RSI message(s) $\mathrm{X}_{\mathrm{R}}$, and the SSI message(s) $\mathrm{X}_{\mathrm{S}}$ (but no other message).
Given $\mathrm{G}$ and  ${\mathrm{Y}_1,\dots,\mathrm{Y}_{P}}$, the user can thus compute $\mathrm{Y}$ and recover $\mathrm{X}_{\mathrm{W}}$ by subtracting off the contribution of $\mathrm{X}_{\mathrm{R}}$ and $\mathrm{X}_{\mathrm{S}}$ from $\mathrm{Y}$. 

\emph{Partition-and-Code based PIR-RSSI scheme:}
For ease of notation, we define ${P:= \lceil K/(M_2+1)\rceil}$. 
First, the user partitions the $K$ message indices ${1,\dots,K}$ into $P$ parts as follows. 
The size of the first ${P-1}$ parts is ${M_2+1}$, and the size of the last part is ${K-(P-1)(M_2+1)}$ ($\leq M_2+1$). 
To partition the $K$ messages indices, the user randomly assigns the demand index $\mathrm{W}$ and the RSI message indices $\mathrm{R}$ to these $P$ parts.
Let $i^{*}\in [P]$ be the index of the part that contains the demand index $\mathrm{W}$.
Let $\alpha_{i^{*}}$ be the size of the part $i^{*}$, and let $r_{i^{*}}$ be the number of RSI message indices that are assigned to the part $i^{*}$. 
Note that ${r_{i^{*}}\geq 0}$ and ${\alpha_{i^{*}}\leq M_2+1}$. 
Next, the user randomly selects ${s_{i^{*}}:= \alpha_{i^{*}}-r_{i^{*}}-1}$ SSI message indices from $\mathrm{S}$, and assigns them to the part $i^{*}$. 
Note that ${s_{i^{*}} = \alpha_{i^{*}}-r_{i^{*}}-1 \leq \alpha_{i^{*}}-1\leq M_2}$.  
Lastly, the user randomly assigns the remaining (not-yet-assigned) message indices to the remaining parts. 
The user then sends the constructed partition (with $P$ parts) to the server as the query. 
The server sends ${\mathrm{Y}_1,\dots,\mathrm{Y}_P}$ back to the user as the answer, where $\mathrm{Y}_i$ is the sum of the messages whose indices are assigned to the part $i$. 

By construction, the combination coefficient vectors corresponding to ${\mathrm{Y}_1,\dots,\mathrm{Y}_P}$ are linearly independent. 
This is because the parts of the partition are mutually disjoint. 
Similarly as in the case of the MDS Code based scheme, it then follows that
${\mathbf{Y}_1,\dots,\mathbf{Y}_P}$ are independent and uniformly distributed over $\mathbbmss{F}_q^n$, and 
consequently, the rate of the Partition-and-Code based scheme is equal to ${1/P}$, or equivalently, ${1/\lceil K/(M_2+1)\rceil}$. 

The Partition-and-Code based scheme satisfies the privacy condition because the index of the demand message ($\mathrm{W}$) and the indices of the RSI messages ($\mathrm{R}$), are randomly assigned to the parts of the partition. 
The recoverability condition is also satisfied because the part $i^{*}$, i.e., the part that contains the demand index $\mathrm{W}$, does not contain any message index not belonging to ${\mathrm{W}\cup\mathrm{R}\cup\mathrm{S}}$. 
In particular, the part $i^{*}$ consists of the demand index, $r_{i^{*}}$ (out of $M_1$) RSI message indices, and $s_{i^{*}}$ (out of $M_2$) SSI message indices. 
Let ${\mathrm{R}_{i^{*}}\subseteq \mathrm{R}}$ and ${\mathrm{S}_{i^{*}}\subseteq \mathrm{S}}$ be the set of the RSI message indices and the set of the SSI message indices thar are assigned to the part $i^{*}$, respectively. 
Then, the user can recover $\mathrm{X}_{\mathrm{W}}$ by subtracting off the contribution of $\mathrm{X}_{\mathrm{R}_{i^{*}}}$ and $\mathrm{X}_{\mathrm{S}_{i^{*}}}$ from $\mathrm{Y}_{i^{*}}$.

\bibliographystyle{IEEEtran}
\bibliography{PIR_PC_Refs}

\end{document}